\documentclass[runningheads,a4paper]{llncs}

\usepackage{amsfonts}  

\title{Interpolation in Valiant's theory}
\institute{LIP\thanks{UMR 5668 ENS Lyon, CNRS, UCBL, INRIA.},
\'Ecole Normale Sup\'erieure de Lyon.\\
{\tt [Pascal.Koiran,Sylvain.Perifel]@ens-lyon.fr}}
\author{Pascal Koiran\and Sylvain Perifel}
\date{\today}

\newcommand\cc{\ensuremath{\mathbb{C}}}

\newcommand\pspace{\ensuremath{\mathsf{PSPACE}}}
\newcommand\vpzero{\ensuremath{\mathsf{VP}^0}}
\newcommand\vp{\ensuremath{\mathsf{VP}}}
\newcommand\vpnb{\ensuremath{\mathsf{VP}_{\mathsf{nb}}}}
\newcommand\vpnbzero{\ensuremath{\mathsf{VP}^0_{\mathsf{nb}}}}
\newcommand\vnp{\ensuremath{\mathsf{VNP}}}
\newcommand\vnpzero{\ensuremath{\mathsf{VNP}^0}}
\newcommand\vnpnb{\ensuremath{\mathsf{VNP}_{\mathsf{nb}}}}
\newcommand\vnpnbzero{\ensuremath{\mathsf{VNP}^0_{\mathsf{nb}}}}
\newcommand\vpip{\ensuremath{\mathsf{V\Pi P}}}
\newcommand\vpipzero{\ensuremath{\mathsf{V\Pi P}^0}}
\newcommand\poly{\ensuremath{\mathsf{poly}}}

\newcommand\nn{\ensuremath{\mathbb{N}}}
\newcommand\fp{\ensuremath{\mathsf{FP}}}
\newcommand\p{\ensuremath{\mathsf{P}}}
\newcommand\pnu{\ensuremath{\mathbb{P}}}
\newcommand\np{\ensuremath{\mathsf{NP}}}
\newcommand\npnu{\ensuremath{\mathbb{NP}}}
\newcommand\per{\ensuremath{\mathrm{PER}}}
\newcommand\sharpp{\ensuremath{\mathsf{\sharp P}}}
\newcommand\pp{\ensuremath{\mathsf{PP}}}
\newcommand\gapp{\ensuremath{\mathsf{GapP}}}
\newcommand\gapppoly{\ensuremath{\mathsf{GapP/poly}}}
\newcommand\chpoly{\ensuremath{\mathsf{CH/poly}}}
\newcommand\ch{\ensuremath{\mathsf{CH}}}

\newcommand\bit{\ensuremath{\mathrm{Bit}}}

\begin{document}

\maketitle

\begin{abstract}
  We investigate the following question: if a polynomial can be evaluated at
  rational points by a polynomial-time boolean algorithm, does it have a
  polynomial-size arithmetic circuit? We argue that this question is certainly
  difficult. Answering it negatively would indeed imply that 
the constant-free versions of the algebraic complexity classes 
 \vp\ and \vnp\ defined by Valiant are different. 
Answering this question positively would imply
  a transfer theorem from 
boolean to algebraic complexity.

  Our proof method relies on Lagrange interpolation and 
on 
recent results connecting the (boolean) counting hierarchy 
to algebraic complexity classes.
As a byproduct we obtain two additional results: 
\begin{itemize}
\item[(i)] The constant-free, degree-unbounded version of Valiant's hypothesis 
$\vp \neq \vnp$ implies the
  degree-bounded version. 
This result was previously known to hold for fields 
of positive characteristic only.
\item[(ii)] If exponential sums of easy to compute
  polynomials can be computed efficiently, then the same is true 
of exponential products. 
We point out an application of this result to the P=NP 
problem in the Blum-Shub-Smale model of computation 
over the field of complex numbers.
\end{itemize}
\end{abstract}

\section{Introduction}

{\bf Motivation --} The starting point of this paper is a question raised 
by Christos Papadimitriou in a personal communication to Erich Kaltofen\footnote{At the Oberwolfach complexity theory workshop where a part of this work was presented in June 2007, several participants told P.K. 
that they had independently thought of the same question.}:
\begin{center}
  \textbf{Question}\hspace{5mm} (*)
\begin{minipage}[c]{.85\textwidth}
  If a multivariate polynomial $P$ can be evaluated by a (boolean)
  polynomial-time algorithm on rational inputs, does that imply that $P$ can
  be computed by a polynomial-size arithmetic circuit?  In such a circuit, the
  only allowed operations are additions, subtractions, and multiplications.
\end{minipage}\medskip
\end{center}

This question can be interpreted in several ways: 
one should at least 
state in which ring the coefficients of $P$ lie, and which constants
can be used by the arithmetic circuits. 
Here we will focus on polynomials with integer coefficients, and most of 
the paper will deal with constant-free circuits. 
In Section~\ref{subsec_constants} we study the case of circuits with rational 
constants computing polynomials with integer coefficients
(as we shall see, this is more natural than it might seem at first sight).

As pointed out by Papadimitriou, Strassen's ``Vermeidung von Divisionen'' (see
for instance~\cite{BCS}, chapter~7) shows that for evaluating a low-degree
polynomial $P$, divisions would not increase exponentially the power of
arithmetic circuits.  It is indeed a natural question whether, more generally,
all boolean operations can be replaced efficiently by additions, subtractions
and multiplications. 
Questions of the same flavour (can ``looking at bits'' help for arithmetic computations?) have been studied before. 
In particular, Kaltofen and Villard have shown that looking at bits 
does help for computing the determinant~\cite{KalVil04}.

{\bf Discussion --} 
It is not clear what the correct answer to question (*) should be.
In this paper we will argue that answering  it either way
seems difficult.  

A natural strategy for obtaining a negative answer to question (*)
 would be to exhibit a family of polynomials that are easy to
evaluate on rational inputs but hard to evaluate by arithmetic circuits.
Unfortunately, there seems to be a lack of candidate polynomials.  Another
difficulty is that a negative answer would imply the separation of the
algebraic complexity classes $\vp^0$ and $\vnp^0$.  This observation is our
main contribution to the study of question (*), and it is established in
Theorem~\ref{cor*}.  The classes $\vp^0$ and $\vnp^0$ are constant-free
versions of the classes $\vp$ (of ``easily computable polynomial families'')
and $\vnp$ (of ``easily definable polynomial families'') 
introduced by Valiant
(precise definitions are given in the next section).  The separation $\vp^0
\neq \vnp^0$ seems very plausible, but it also seems very difficult to
establish.
As explained at the beginning of the introduction, 
we study in Section~\ref{subsec_constants} 
 the case of circuits with rational 
constants computing polynomials with integer coefficients.
Allowing rational constants makes the hypothesis that question (*) 
has a negative answer stronger than in the constant-free case.
Accordingly, we obtain a stronger conclusion: we can now show that
the hypothesis would imply a superpolynomial lower bound 
on the size of arithmetic circuits computing the permanent.

Obtaining a positive answer to question (*) 
also seems difficult since it would imply the following transfer
theorem: $\fp = \sharpp \Rightarrow \vp = \vnp$ (assuming that $\fp =
\sharpp$, the permanent must be in $\fp$; a positive answer to question (*)
would therefore imply that the permanent is in $\vp$, and that $\vp = \vnp$ by
completeness of the permanent).  Unfortunately, in spite of all the work
establishing close connections between the boolean model of computation and
the algebraic models of Valiant and of Blum, Shub and
Smale~\cite{BCSS,burgisser2000,Burg04,FouKoi98,FouKoi00,KoiPe06,KoiPe07b,KoiPe07a}
no such transfer theorem is known.  In fact, we do not know of any hypothesis
from boolean complexity theory that would imply the equality $\vp = \vnp$ (but
transfer theorems in the opposite direction were established
in~\cite{burgisser2000}).

{\bf Summary of results --} 
Most of our results are derived from Theorem~\ref{main_th}, our main theorem:
if the evaluation of 
a family of polynomials $(f_n)$ at integer points is a problem 
that lies in the (non uniform) counting hierarchy,
 the hypothesis $\vp^0 = \vnp^0$ implies that $(f_n)$ can be evaluated by 
polynomial-size arithmetic circuits.
Theorem~\ref{cor*}, which contains our main contribution to the study 
of question~(*), follows immediately since polynomial-time problems 
lie in the counting hierarchy.
The proof of Theorem~\ref{main_th} relies on techniques 
from~\cite{ABKM2005,burgisser2007} and on Lagrange interpolation.
Besides the application to question (*), we derive two additional
results from Theorem~\ref{main_th}:
\begin{itemize}
\item The elements of the complexity classes $\vp$, $\vnp$ and of their
  constant-free versions are families of polynomials of polynomially bounded
  degree.  We show in Theorem~\ref{cor_bnb} that the collapse $\vp^0 =
  \vnp^0$ would imply the same collapse for the unbounded versions of $\vp^0$
  and $\vnp^0$. For fields of positive characteristic, the same result (and
  its converse) was obtained with different techniques by
  Malod~\cite{Malod07,malod2003}.

\item Our third application of Theorem~\ref{main_th} is to the ``$\p = \np$?''
  problem in the Blum-Shub-Smale model of computation over $\cc$
(or more generally, fields of characteristic 0). One natural
  strategy for separating $\p_{\cc}$ from $\np_{\cc}$ would be to exhibit a
  problem $A$ in $\np_{\cc} \setminus \p_{\cc}$.  Drawing on results
  from~\cite{KoiPe06}, we show that this strategy is bound to fail for a
  fairly large class of ``simple'' problems~$A$, unless one can prove that
  $\vp^0 \neq \vnp^0$.  The class of ``simple'' problems that we have in mind
  is $\np_{(\cc,+,-,=)}$.  This is the class of \np\ problems over the set of
  complex numbers endowed with addition, subtraction, and equality tests
  (there is therefore no multiplication in this structure).  It contains many
  natural problems, such as Subset Sum and Twenty Questions~\cite{BCSS,SS95},
  that most likely belong to $\np_{\cc} \setminus \p_{\cc}$.
As an intermediate result, we show in Theorem~\ref{cor_bnb} that 
if exponential sums of easy to compute
  polynomials can be computed efficiently, then the same is true 
of exponential products. 
\end{itemize}




\section{Preliminaries}

\subsection{Valiant's Classes} \label{valiant}

In Valiant's model, one computes families of polynomials. A book-length
treatment of this topic can be found in~\cite{burgisser2000}. We fix a field
$K$ of characteristic zero.

An arithmetic circuit is a circuit whose inputs are indeterminates
$x_1,\dots,x_{u(n)}$ together with arbitrary constants of $K$; there are $+$,
$-$ and $\times$-gates, and we therefore compute multivariate polynomials. The
polynomial computed by an arithmetic circuit is defined in the usual way by
the polynomial computed by its output gate. The size of a circuit is the
number of gates.

Thus a family $(C_n)$ of arithmetic circuits computes a family $(f_n)$ of
polynomials, $f_n\in K[x_1,\dots,x_{u(n)}]$. The class \vpnb\ defined in
\cite{malod2003} is the set of families $(f_n)$ of polynomials computed by a
family $(C_n)$ of polynomial-size arithmetic circuits, i.e., $C_n$ computes
$f_n$ and there exists a polynomial $p(n)$ such that $|C_n|\leq p(n)$ for all
$n$. We will assume without loss of generality that the number $u(n)$ of
variables is bounded by a polynomial function of $n$. The subscript
``$\mathsf{nb}$'' indicates that there is no bound on the degree of the
polynomial, in contrast with the original class \vp\ of Valiant where a
polynomial bound on the degree of the polynomial computed by the circuit is
required. Note that these definitions are nonuniform.

The class \vnp\ is the set of families of polynomials defined by an
exponential sum of \vp\ families. More precisely, $(f_n(\bar x))\in\vnp$ if
there exists $(g_n(\bar x,\bar y))\in\vp$ and a polynomial $p$ such that
$|\bar y| = p(n)$ and $f_n(\bar x) =
\sum_{\bar\epsilon\in\{0,1\}^{p(n)}}g_n(\bar x, \bar\epsilon).$ Similarly, the
class \vpip\ is the set of families of polynomials defined by an exponential
product of \vpnb\ families. More precisely, $(f_n(\bar x))\in\vpip$ if there
exists $(g_n(\bar x,\bar y))\in\vpnb$ and a polynomial $p$ such that $|\bar y|
= p(n)$ and $f_n(\bar x) = \prod_{\bar\epsilon\in\{0,1\}^{p(n)}}g_n(\bar x,
\bar\epsilon).$

We can also define constant-free circuits: the only constant allowed is then 1
(in order to allow the computation of constant polynomials). In this case, we
compute polynomials with integer coefficients. If $f$ is a polynomial with
integer coefficients, we denote by $\tau(f)$ the size of a smallest
constant-free circuit computing $f$. For classes of families of polynomials,
we will use the superscript 0 to indicate the absence of constant: for
instance, we will write \vpnbzero. For bounded-degree classes, we are to be
more careful because we also want to avoid the computation of constants of
exponential bitsize: we first need the following definition.
\begin{definition}
  Let $C$ be an arithmetic circuit. The \emph{formal degree} of a gate of $C$
  is defined by induction:
  \begin{itemize}
  \item the formal degree of an input is 1;
  \item the formal degree of a gate $+$ or $-$ is the maximum of the formal
    degrees of its inputs;
  \item the formal degree of a gate $\times$ is the sum of the formal degrees
    of its inputs.
  \end{itemize}
  Now, the formal degree of a circuit is the formal degree of the output gate.
\end{definition}
We are now able to define constant-free degree-bounded Valiant's classes. 
A family of polynomials $(f_n)$ belongs to \vpzero\ if it is computable
by a family of circuits of size and formal degree bounded 
by a polynomial function of $n$. 
The class \vnpzero\ is then defined accordingly by a
sum of \vpzero\ families,
in the same way as $\vnp$ is defined from $\vp$.

\begin{remark}\label{remark1}
  The hypothesis $\tau(\per_n)=n^{O(1)}$ used in~\cite{burgisser2007} is
  implied by the hypothesis 
$\vnpzero\subset\vpnbzero$ and hence by $\vpzero=\vnpzero$. As
  mentioned in~\cite{burgisser2007}, the converse $\tau(\per_n)=
  n^{O(1)}\Rightarrow\vnpzero=\vpzero$ is not known to hold, because the
  family $(\per_n)$ is not known to be \vnpzero-complete in a constant-free
  context (the proof of completeness of Valiant~\cite{valiant1979} indeed uses
  the constant $1/2$). We will mostly be concerned by the hypothesis
  $\vpzero=\vnpzero$, but we will come to the hypothesis
  $\tau(\per_n)=n^{O(1)}$ in Section~\ref{subsec_constants} when dealing with
  circuits with constants.
\end{remark}

\subsection{Counting Classes}

In this paper we will encounter several counting classes, in particular the
counting hierarchy defined below. Let us first see two classes of functions,
\sharpp\ and \gapp.
\begin{definition}
  \begin{itemize}
  \item The class \sharpp\ is the set of functions
    $f:\{0,1\}^*\rightarrow\{0,1\}^*$ such that there exist a language
    $A\in\p$ and a polynomial $p(n)$
    satisfying $$f(x)=\#\{y\in\{0,1\}^{p(|x|)}:(x,y)\in A\}.$$
  \item A function $f$ is in \gapp\ if it is the difference of two functions
    in \sharpp.
  \end{itemize}
\end{definition}
Returning to classes of languages we 
recall the definition of the counting hierarchy, introduced by
Wagner~\cite{wagner1986}. It contains all the polynomial hierarchy
$\mathsf{PH}$ and is contained in \pspace. It is defined via the ``majority''
operator $\mathbf{C}$ as follows.
\begin{definition}
  \begin{itemize}
  \item If $K$ is a complexity class, the class $\mathbf{C}.K$ is the set of
    languages $A$ such that there exist a language $B\in K$ and a polynomial
    $p(n)$ satisfying
    $$x\in A\iff \#\{y\in\{0,1\}^{p(|x|)}:(x,y)\in B\}\geq 2^{p(|x|)-1}.$$
  \item The $i$-th level $\mathsf{C_iP}$ of the counting hierarchy is defined
    recursively by $\mathsf{C_0P}=\p$ and
    $\mathsf{C_{i+1}P}=\mathbf{C}.\mathsf{C_iP}$. The counting hierarchy \ch\
    is the union of all these levels $\mathsf{C_iP}$.
  \end{itemize}
\end{definition}
Level 1 of \ch, that is, $\mathbf{C}.P$, is also called \pp. Since Valiant's
classes are nonuniform, we will rather work with nonuniform versions of these
boolean classes, as defined now following Karp and Lipton~\cite{KL1982}.
\begin{definition}
  If $K$ is a complexity class, the class $K/\poly$ is the set of languages
  $A$ such that there exist a language $B\in K$, a polynomial $p(n)$ and a
  family of words (called advices) $(a_n)_{n\geq 0}$ satisfying
  \begin{itemize}
  \item for all $n\geq 0$, $|a_n|\leq p(n)$;
  \item for all word $x$, $x\in A\iff (x,a(|x|))\in B$.
  \end{itemize}
  Remark that the advice only depends on the size of $x$: it must therefore be
  the same for all words of same length.
\end{definition}

\subsection{Sequences of Integers}

Our aim now is to introduce a notion of complexity of a sequence of
integers. In order to avoid dealing with the sign of integers separately, we
assume that we can retrieve it from the boolean encoding of the integers. For
example, 
the sign could be given by the first bit of the encoding 
and the absolute value by the remaining bits.

\begin{definition}
  A \emph{sequence of exponential bitsize} is a sequence of integers
  $(a(n,k))$ such that there exists a polynomial $p(n)$ satisfying:
  \begin{enumerate}
  \item $a(n,k)$ is defined for $n,k\in\nn$ and $0\leq k < 2^{p(n)}$;
  \item for all $n>1$, for all $k < 2^{p(n)}$, the bitsize of $a(n,k)$ is
    $\leq 2^{p(n)}$.
  \end{enumerate}
  From $a(n,k)$, the following language is then defined:
  $$\bit(a)=\{(1^n,k,j,b)|\mbox{ the $j$-th bit of $a(n,k)$ is }b\},$$
\end{definition}
The reader should
be aware that the above definition and the next one are not quite the same 
as in~\cite{burgisser2007}: we use a unary encoding for $n$ instead of a binary encoding.

\begin{definition}
  A sequence $a(n,k)$ of exponential bitsize is \emph{definable in \chpoly} if
  the language $\bit(a)$ is in \chpoly.
\end{definition}
\begin{remark}
  We shall also meet sequences with more than two parameters $(n,k)$, for
  example $a(n,\alpha^{(1)},\dots,\alpha^{(n)})$ for some integers
  $\alpha^{(i)}$. In order to see it as a sequence with two parameters,
  $(\alpha^{(1)},\dots,\alpha^{(n)})$ will be considered as the encoding of a
  single integer. The parameter $n$ might also be given as a subscript, as in
  $f_n(k)$, which should better be written $f(n,k)$.
\end{remark}

Let us now propose a similar definition for families of polynomials.
\begin{definition} \label{eval}
  Let $(f_n(x_1,\dots,x_{u(n)}))$ be a family of polynomials with integer
  coefficients. We say that $(f_n)$ can be \emph{evaluated in \chpoly\ at
    integer points} if the following conditions are satisfied:
  \begin{enumerate}
  \item the number $u(n)$ of variables is polynomially bounded;
  \item the degree of $f_n$ as well as the bitsize of the coefficients of
    $f_n$ are bounded by $2^{p(n)}$ for some polynomial $p(n)$;
  \item the language $\{(1^n,i_1,\dots,i_{u(n)},j,b)|\mbox{ the
    }j\mbox{-th bit of }f_n(i_1,\dots,i_{u(n)})\mbox{ is }b\}$ is in \chpoly.
  \end{enumerate}
\end{definition}
\begin{remark} \label{peval}
The same definition can be made for other complexity classes than \chpoly.
For instance, if we replace \chpoly\ by \p\ we obtain the notion of ``polynomial time evaluation at integer points''. This notion will be useful for the study 
of question~(*). 
\end{remark}
The following lemma is obvious from these definitions.
\begin{lemma}\label{lem_poly_int}
  The family $(f_n(x_1,\dots,x_{u(n)}))$ can be evaluated in \chpoly\ at
  integer points if and only if the sequence of integers
  $a(n,i_1,\dots,i_{u(n)})=f_n(i_1,\dots,i_{u(n)})$ is definable in \chpoly.
\end{lemma}
The following theorem of~\cite[Theorem~4.1]{ABKM2005} will also be useful due
to its corollary below.
\begin{theorem}
  Let BitSLP be the following problem: given a constant-free arithmetic
  circuit computing an integer $N$, and given $i\in\nn$ in binary, decide
  whether the $i$-th bit of the binary representation of $N$ is
  1. Then BitSLP is in \ch.
\end{theorem}
\begin{corollary}\label{cor_vp}
  If $(f_n)\in\vpnbzero$ then it can be evaluated in \chpoly\ at integer
  points.
\end{corollary}

The results of this paper rely on the following link between Valiant's
classes and the counting hierarchy, \cite[Lemmas~2.5 and~2.12]{burgisser2007}.
\begin{lemma}\label{lemma_ch}
  If $\vpzero=\vnpzero$ then $\chpoly=\p/\poly$.
\end{lemma}
In particular, Lemma~\ref{lemma_ch} was used to 
show that big sums and products are computable in the counting hierarchy,
\cite[Theorem~3.7]{burgisser2007}. As already mentioned, the context is not
exactly the same as in~\cite{burgisser2007} because 
we use a unary encoding.
We now give a version of this result which is 
just an easy ``scaling up''
of~\cite[Theorem~3.7]{burgisser2007} (it is enough to define
$a'(2^{p(n)},k)=a(n,k)$ and to apply the result of B\"urgisser).
\begin{theorem}\label{th_burg}
Let $p(n)$ be a polynomial and suppose $a=(a(n,k))_{n\in\nn,k\leq
      2^{p(n)}}$ is definable in \chpoly. Consider
the sequences
    $$b(n)=\sum_{k=0}^{2^{p(n)}}a(n,k)\mbox{ and
    }d(n)=\prod_{k=0}^{2^{p(n)}}a(n,k).$$
    Then $(b(n))_{n\in\nn}$ and $(d(n))_{n\in\nn}$ are definable in \chpoly.
  
    Suppose now that $(s(n))_{n\in\nn}$ and $(t(n))_{n\in\nn}$ are definable
    in \chpoly. Then the sequence of products $(s(n)t(n))_{n\in\nn}$, and, if
    $t(n)>0$, the sequence of quotients $(\lfloor
    s(n)/t(n)\rfloor)_{n\in\nn}$, are definable in \chpoly.
\end{theorem}

\section{Interpolation}

We now begin the main technical developments.

\subsection{Coefficients}

The following lemma is Valiant's criterion~\cite{valiant1979}, see
also~\cite[Prop.~2.20]{burgisser2000} and~\cite[Th.~2.3]{koiran2004}.
\begin{lemma}\label{lemma_criterion}
  Let $a:(1^n,i)\mapsto a(1^n,i)$ be a function of \gapppoly, where $n$ is
  given in unary and $i$ in binary. Let $p(n)$ be a polynomial and define the
  following sequence of polynomials:
  $$f_n(x_1,\dots,x_{p(n)})=\sum_{i=0}^{2^{p(n)}-1}a(1^n,i)x_1^{i_1}\cdots
  x_{p(n)}^{i_{p(n)}},$$
  where $i_j$ is the $j$-th bit in the binary expression of $i$.

  Then $(f_n)\in\vnpzero$.
\end{lemma}

Here is a ``scaled up'' generalization of~\cite[Th.~4.1(2)]{burgisser2007} to
multivariate polynomials.
\begin{lemma}\label{lemma_generalization}
  Let
  $$f_n(x_1,\dots,x_n)=\sum_{\alpha^{(1)},\dots,\alpha^{(n)}}
  a(n,\alpha^{(1)},\dots,\alpha^{(n)})x_1^{\alpha^{(1)}}\cdots
  x_n^{\alpha^{(n)}},$$
  where the integers $\alpha^{(i)}$ range from 0 to $2^n-1$ and
  $a(n,\alpha^{(1)},\dots,\alpha^{(n)})$ is a sequence of integers of absolute
  value $< 2^{2^n}$ definable in \chpoly.

  If $\vpzero=\vnpzero$ then $(f_n)\in\vpnbzero$.
\end{lemma}

\begin{proof}
  Expand $a$ in binary: $a(n,\alpha^{(1)},\dots,\alpha^{(n)})=\sum_{i=0}^{2^n}
  a_i(n,\bar\alpha)2^i$. Let $h_n$ be the following polynomial:
  $$h_n(x_{1,1},x_{1,2},\dots,x_{1,n},x_{2,1},\dots,x_{n,n},z_1,\dots,z_n)=$$
  $$\sum_{i=0}^{2^n}\sum_{\bar\alpha}a_i(n,\bar\alpha) z_1^{i_1}\cdots
  z_n^{i_n} x_{1,1}^{\alpha^{(1)}_1}x_{1,2}^{\alpha^{(1)}_2}\cdots
  x_{1,n}^{\alpha^{(1)}_n}x_{2,1}^{\alpha^{(2)}_1}\cdots
  x_{n,n}^{\alpha^{(n)}_n}.$$
  Then we have:
  $$h_n(x_1^{2^0},x_1^{2^1},\dots,x_1^{2^n},x_2^{2^0},\dots,
  x_n^{2^n},2^{2^0},2^{2^1},\dots,2^{2^n})=f_n(x_1,\dots,x_n).$$ Since
  $\vpzero=\vnpzero$, by Lemma~\ref{lemma_ch} the nonuniform counting
  hierarchy collapses, therefore computing the $i$-th bit $a_i(n,\bar\alpha)$
  of $a(n,\bar \alpha)$ on input $(1^n,\bar \alpha,i)$ is in \gapppoly\ (and
  even in \p/\poly). By Lemma~\ref{lemma_criterion}, $(h_n)\in\vnpzero$. By
  the hypothesis $\vpzero=\vnpzero$, $(h_n)\in\vpzero$ and thus using repeated
  squaring for computing big powers yields $(f_n)\in\vnpnbzero$.\qed
\end{proof}

\subsection{Interpolation}

Let us now state two lemmas on interpolation polynomials.
\begin{lemma}[multivariate Lagrange interpolation]\label{lemma_interpolation}
  Let $p(x_1,\dots,x_n)$ be a polynomial of degree $\leq d$. Then
  $$p(x_1,\dots,x_n)=\sum_{0\leq i_1,\dots,i_n\leq d}p(i_1,\dots,i_n)
  \prod_{k=1}^n\biggl(\prod_{j_k\neq i_k}\frac{x_k-j_k}{i_k-j_k}\biggr),$$
  where the integers $j_k$ range from 0 to $d$.
\end{lemma}

\begin{proof}
  The proof goes by induction on the number $n$ of variables. For $n=1$, this
  is the usual Lagrange interpolation formula: we have
  $$p(x)=\sum_{i=0}^dp(i)\prod_{j\neq i}\frac{x-j}{i-j}$$
  because both polynomials are of degree $\leq d$ and coincide on at least
  $d+1$ distinct points.

  For $n+1$, the induction case $n=1$ yields
  $$p(x_1,\dots,x_{n+1})=\sum_{i_{n+1}=0}^dp(x_1,\dots,x_n,i_{n+1})\prod_{j_{n+1}\neq
    i_{n+1}}\frac{x_{n+1}-j_{n+1}}{i_{n+1}-j_{n+1}}.$$
  By induction hypothesis, this is equal to
  $$\sum_{i_{n+1}=0}^d\left(\sum_{0\leq i_1,\dots,i_n\leq d}p(i_1,\dots,i_n)
  \prod_{k=1}^n\prod_{j_k\neq i_k}\frac{x_k-j_k}{i_k-j_k}\right)
  \prod_{j_{n+1}\neq i_{n+1}}\frac{x_{n+1}-j_{n+1}}{i_{n+1}-j_{n+1}}$$
  which is the desired result.\qed
\end{proof}




\begin{lemma}\label{lemma_inter_ch}
  Let $a(n)=\prod_{i=0}^{2^n-1}\prod_{j\neq i}(i-j)$, where $j$ ranges from 0
  to $2^n-1$. Let $p_{i_1,\dots,i_n}(\bar x)$ be the following family of
  polynomials:
  $$p_{i_1,\dots,i_n}(x_1,\dots,x_n)=\prod_{k=1}^n
  \biggl(a(n)\prod_{j_k\neq i_k}\frac{x_k-j_k}{i_k-j_k}\biggr),$$
  where the integers $j_k$ range from 0 to $2^n-1$ and the integers $i_k$ are
  given in binary and range from 0 to $2^n-1$. Then the coefficients of
  $p_{i_1,\dots,i_n}$ are integers definable in the counting hierarchy, as is
  $a(n)$.
\end{lemma}

\begin{proof}
  As a first step, 
note that the coefficient of the monomial
  $x_1^{\alpha_1}\cdots x_n^{\alpha_n}$ in $p_n$ is equal to the product of
  the coefficients of the monomials $x_k^{\alpha_k}$ in the univariate
  polynomials $a(n)\prod_{j_k\neq i_k}\frac{x_k-j_k}{i_k-j_k}$. Hence we just
  have to check that these different coefficients of univariate polynomials
  are themselves definable in the counting hierarchy. Let us first focus on
  the univariate polynomial $\prod_{j_k\neq i_k}(x_k-j_k)$, that is, let us
  forget the multiplicative term $b(n,i_k)=a(n)/\prod_{j_k\neq i_k}(i_k-j_k)$
  for the moment.

  We use the same argument as~\cite[Cor.~3.9]{burgisser2007}. Namely, 
we 
remark that the coefficients of this polynomial are bounded in absolute value 
by $2^{2^{n^2}}$. Therefore in the univariate polynomial $\prod_{j_k\neq
    i_k}(x_k-j_k)$ we can replace the variable $x_k$ by $2^{2^{n^2}}$ and
  there will be no overlap of the coefficients of the different powers of
  $x_k$, thus we can recover the coefficients of the monomial from the value
  of this product. By the first 
part of Theorem~\ref{th_burg}, we can evaluate
  in the counting hierarchy the polynomial at the point $2^{2^{n^2}}$, because
  it is a product of exponential size. So the coefficients are definable in
  the counting hierarchy.

  It is now enough to 
note that the first 
part of Theorem~\ref{th_burg}
  implies that $a(n)$ as well as $b(n,i_k)=a(n)/\prod_{j_k\neq i_k}(i_k-j_k)$
  are also definable in the counting hierarchy.\qed
\end{proof}

Remark that the sequence $a(n)$ of Lemma~\ref{lemma_inter_ch} is introduced
only so as to obtain integer coefficients. We will then divide by $a(n)$ in the
next proofs.

\subsection{Main Results}

Let us now state the main theorem.
\begin{theorem}\label{main_th}
  Let $(f_n(x_1,\dots,x_{u(n)}))$ be a family of multivariate
  polynomials. Suppose $(f_n)$ can be evaluated in \chpoly\ at integer
  points. If $\vpzero=\vnpzero$ then $(f_n)\in\vpnbzero$.
\end{theorem}

\begin{proof}
  The goal is to 
use the interpolation formula of Lemma~\ref{lemma_interpolation}:
  \begin{equation} \label{fandb}
f_n(x_1,\dots,x_{u(n)})=\sum_{0\leq i_1,\dots,i_{u(n)}\leq
    d}b_{i_1,\dots,i_{u(n)}}(\bar x),
\end{equation}
  where $b_{i_1,\dots,i_{u(n)}}(\bar x)=f_n(i_1,\dots,i_{u(n)})
  \prod_{k=1}^{u(n)}\prod_{j_k\neq i_k}\frac{x_k-j_k}{i_k-j_k}$.
We will show that the coefficients of $b_{i_1,\dots,i_{u(n)}}$ and $f_n$
are definable in \chpoly. The conclusion of the theorem will then 
follow from  Lemma~\ref{lemma_generalization}.

In order to show that the coefficients of $b_{i_1,\dots,i_{u(n)}}$
are definable in \chpoly, we note that 
the polynomial
  $p_{i_1,\dots,i_n}$ and the sequence $a(n)$ of Lemma~\ref{lemma_inter_ch}
  satisfy the relation
  $$b_{i_1,\dots,i_{u(n)}}(\bar x) =
  a(u(n))^{-u(n)}f_n(i_1,\dots,i_{u(n)})p_{i_1,\dots,i_{u(n)}}(\bar x).$$
By Lemma~\ref{lemma_inter_ch}, the coefficients of
  $p_{i_1,\dots,i_{u(n)}}(\bar x)$ are definable in \ch.
By hypothesis, $(f_n)$ can be evaluated in \chpoly\ at integer points. This
  implies by Lemma~\ref{lem_poly_int} that $f_n(i_1,\dots,i_{u(n)})$ is
  definable in \chpoly. 
   This is also the case
  of the product $f_n(i_1,\dots,i_{u(n)})p_{i_1,\dots,i_{u(n)}}(\bar x)$ by
  Theorem~\ref{th_burg}. Now, the same theorem enables us to divide by
  $a(u(n))^{u(n)}$, thereby showing that the coefficients of
  $b_{i_1,\dots,i_{u(n)}}(\bar x)$ are definable in \chpoly. 
It then follows from~(\ref{fandb}) and another application of
  Theorem~\ref{th_burg} that the coefficients of $f_n$
  are definable in \chpoly. Therefore by
  Lemma~\ref{lemma_generalization}, $(f_n)\in\vpnbzero$ under the hypothesis
  $\vpzero=\vnpzero$.\qed
\end{proof}

We now derive some consequences of Theorem~\ref{main_th}.
\begin{theorem}
  Let $(f_n(\bar x,\bar\epsilon))\in\vpnbzero$. Let
  $$g_n(\bar x)=\sum_{\bar\epsilon}f_n(\bar x,\bar\epsilon)\mbox{ and }
  h_n(\bar x)=\prod_{\bar\epsilon}f_n(\bar x,\bar\epsilon).$$
  If $\vnpzero=\vpzero$ then $(g_n)$ and $(h_n)$ are in \vpnbzero.
\end{theorem}

\begin{proof}
  By Corollary~\ref{cor_vp} $(f_n)$ can be evaluated in \chpoly\ at integer
  points. Now, using Lemma~\ref{lem_poly_int} before and after the first 
 part  of Theorem~\ref{th_burg} shows that $(g_n)$ and $(h_n)$ can also be
  evaluated in \chpoly\ at integer points. The result then follows by
  Theorem~\ref{main_th}.\qed



\end{proof}

The following is now immediate.
\begin{theorem}\label{cor_bnb}
The hypothesis  $\vpzero=\vnpzero$ implies that $\vpnbzero=\vnpnbzero$
and $\vpnbzero=\vpipzero$.
\end{theorem}

\begin{remark}
  It is not clear whether the converse of the first implication 
in Theorem~\ref{cor_bnb}
  ($\vpzero=\vnpzero \Longrightarrow\vpnbzero=\vnpnbzero$) holds true. 
This is related to the issue of large constants in arithmetic
circuits: it seems difficult to rule out the possibility that 
some polynomial family in $\vnp^0$  
(for instance, the permanent or the hamiltonian) does not lie in $\vp^0$
but is still computable by polynomial-size arithmetic circuits 
using integer constants
of exponential bit size.

The converse does hold if arbitrary constants are allowed: we indeed have
  $\vpnb=\vnpnb\Longrightarrow\vp=\vnp$. But in this non-constant-free
  context, it is not clear whether $\vp=\vnp$ unconditionally implies
  $\vpnb=\vnpnb$: indeed, in this context the generalized Riemann hypothesis
  would be needed to make the proof of Lemma~\ref{lemma_ch} work
(see~\cite{burgisser2007} for details).
\end{remark}

As mentioned in the introduction, another corollary concerns a transfer
theorem with classes of algebraic complexity in the BSS model. Blum, Shub and
Smale~\cite{BCSS,BSS1989} have defined the classes \p\ and \np\ over the real
and complex fields. It was extended to arbitrary structures by
Poizat~\cite{poizat1995}. Here we use nonuniform versions of these classes,
hence the notations $\pnu$ and $\npnu$.

Theorem~\ref{cor_transfer} below proves that, over a field of characteristic
zero, if we separate (the nonuniform versions of) \p\ and \np\ thanks to a
``simple'' \np\ problem, then we separate (the constant-free versions of) \vp\
and \vnp. The class of ``simple'' problems here is \np\ where the
multiplication is not allowed, i.e., the only operations are $+,-$ and
$=$. It contains in particular Twenty Questions and Subset Sum. 
We will need a result from~\cite{KoiPe06}: 
\begin{theorem}
  Let $K$ be a field of characteristic zero. If $\vpnbzero=\vpipzero$ then
  $\npnu_{(K,+,-,=)}\subseteq\pnu_{(K,+,-,\times,=)}$.
\end{theorem}
By Theorem~\ref{cor_bnb}, the following is immediate.
\begin{theorem}\label{cor_transfer}
  Let $K$ be a field of characteristic zero. If $\vpzero=\vnpzero$ then
  $\npnu_{(K,+,-,=)}\subseteq\pnu_{(K,+,-,\times,=)}$.
\end{theorem}

At last, as a corollary of Theorem~\ref{main_th} again, we obtain the
following result concerning question (*), 
suggesting
that it will be hard to refute.
As pointed out in the introduction, this result does
  not give any evidence concerning the answer to question (*) since the
  separation $\vpzero\neq\vnpzero$ is very likely to be true.
\begin{theorem}\label{cor*}
If question (*) has a negative answer then
  $\vpzero\neq\vnpzero$.
More precisely, let $(f_n)$ be a family of multivariate polynomials
which can be evaluated
in polynomial time at integer points (in the sense of Remark~\ref{peval}).
If $\vpzero=\vnpzero$ then $(f_n)\in\vpnbzero$.
\end{theorem}

\subsection{Arithmetic Circuits with Constants}
\label{subsec_constants}

In this section
we investigate another interpretation of question (*):
we still consider polynomials with integer coefficients, but we allow
rational constants in our circuits
(it turns out that the constant $1/2$ plays a special role due to its appearance in the completeness proof for the permanent).
The hypothesis that question (*) has a negative answer is then stronger,
and we obtain a stronger conclusion than in Theorem~\ref{cor*}.
Namely, we can conclude that
$\tau(\per_n)\neq n^{O(1)}$ instead of $\vnpzero\neq \vpzero$ (see
Remark~\ref{remark1}).
We recall that $\tau$, the constant-free arithmetic circuit complexity of a polynomial, is defined in Section~\ref{valiant}.
\begin{theorem}\label{cor_constant}
 As explained above,  we consider here polynomials with integer coefficients
but circuits with rational constants.
  If question (*) has a negative answer, then
  $\tau(\per_n)$ is not polynomially bounded.

More precisely, let $(f_n)$ be a family of multivariate polynomials
which can be evaluated
in polynomial time at integer points (in the sense of Remark~\ref{peval}).
If $\tau(\per_n)$ is polynomially bounded, $(f_n)$ can be evaluated by a
family of polynomial-size arithmetic circuits that use only the constant~$1/2$.
\end{theorem}

It is easy to see that Theorem~\ref{cor_constant} follows from a slight
modification of the different lemmas above. Lemma~\ref{lemma_ch} is replaced
by the following
stronger lemma, 
from~\cite{burgisser2007}.
\begin{lemma}
  If $\tau(\per_n)=n^{O(1)}$ then $\chpoly=\p/\poly$.
\end{lemma}
Then Lemma~\ref{lemma_generalization} is replaced by the following result,
whose proof relies on an inspection of Valiant's proof~\cite{valiant1979} of
\vnp-completeness of the permanent, see~\cite{burgisser2007}.
\begin{lemma}
  Let
  $$f_n(x_1,\dots,x_n)=\sum_{\alpha^{(1)},\dots,\alpha^{(n)}}
  a(n,\alpha^{(1)},\dots,\alpha^{(n)})x_1^{\alpha^{(1)}}\cdots
  x_n^{\alpha^{(n)}},$$
  where the integers $\alpha^{(i)}$ range from 0 to $2^n-1$ and
  $a(n,\alpha^{(1)},\dots,\alpha^{(n)})$ is a sequence of integers of absolute
  value $< 2^{2^n}$ definable in \chpoly.

  If $\tau(\per_n)=n^{O(1)}$ then there exists a polynomial $p(n)$ such that
  $\tau(2^{p(n)}f_n)=n^{O(1)}$.
\end{lemma}
Finally, Theorem~\ref{main_th} becomes the following.
\begin{lemma}
  Let $(f_n(x_1,\dots,x_{u(n)}))$ be a family of multivariate polynomials
  (with integer coefficients). Suppose $(f_n)$ can be evaluated in \chpoly\ at
  integer points. If $\tau(\per_n)=n^{O(1)}$ then there exists a polynomial
  $p(n)$ such that $\tau(2^{p(n)}f_n)=n^{O(1)}$.
\end{lemma}
Theorem~\ref{cor_constant} follows since the coefficient $2^{p(n)}$ can be
cancelled by multiplying by the constant $2^{-p(n)}$, 
which can be computed from scratch from the constant $1/2$.

\small

\section*{Acknowledgments}
We would like to thank Erich Kaltofen and Christos Papadimitriou 
for sharing their thoughts on question (*).


\end{document}